\documentclass[11pt]{amsart}
\usepackage{geometry}                
\usepackage{amssymb}

\usepackage{amsthm}
\usepackage{bm}
\usepackage{mathrsfs}
\usepackage[colorlinks=true]{hyperref}

\theoremstyle{plain}
\newtheorem{theorem}{Theorem}
\newtheorem{lemma}[theorem]{Lemma} 
\newtheorem{proposition}[theorem]{Proposition}

\theoremstyle{definition}
\newtheorem{definition}{Definition}

\theoremstyle{remark}
\newtheorem{remark}{Remark}

\title[Self-adjoint extensions  and  unitary operators on the boundary]{Self-adjoint extensions \\ and  unitary operators on the boundary}

\author[P. Facchi]{Paolo Facchi}
\address[P. Facchi]{Dipartimento di Fisica and MECENAS, Universit\`a di Bari
\and INFN Sezione di Bari, I-70126 Bari, Italy}
\email{paolo.facchi@ba.infn.it} 
\author[G. Garnero]{Giancarlo Garnero} 
\address[G. Garnero]{Dipartimento di Fisica and MECENAS, Universit\`a di Bari
\and INFN Sezione di Bari, I-70126 Bari, Italy}
\email{giancarlo.garnero@ba.infn.it} 
\author[M. Ligab\`o]{Marilena Ligab\`o}
\address[M. Ligab\`o]{Dipartimento di Matematica, Universit\`a di Bari, I-70125  Bari, Italy} 
\email{marilena.ligabo@uniba.it}



\date{\today}

\begin{document}

\begin{abstract}
We establish a bijection between the self-adjoint extensions of the Laplace operator on a bounded regular domain and the unitary operators on the boundary. Each unitary encodes a specific relation between the boundary value of the function and its normal derivative.  
This bijection sets up a characterization of all physically admissible  dynamics of a nonrelativistic quantum particle confined in a cavity.
Moreover, this correspondence is discussed also at the level of quadratic forms. Finally, the connection between this  parametrization of the extensions and the classical one, in terms of boundary self-adjoint operators on closed subspaces, is shown. 

\end{abstract}

\maketitle

\section{Introduction}

In the last few years there has been an increasing interest in the physics of quantum systems confined in a bounded spatial region and in the prominent role of quantum boundary conditions. It has been realized that the presence of boundaries can often catalyze and amplify the genuine quantum behavior of a system. Several  examples are worth mentioning, ranging from the Aharonov-Bohm effect in quantum mechanics, to the quantum Hall effect in solid state physics, from  anomalies and the Casimir effect in quantum field theory, to fluctuating topologies in quantum gravity. For a review, see e.g.~\cite{marmo,reviewbal,bangalectures}.

All physical dynamics of closed quantum systems are implemented by strongly continuous one-parameter unitary groups, which by Stone's theorem are in  one to one correspondence with their generators that must be self-adjoint operators. See e.g.~\cite{RS}. Physically, the generator is the system's Hamiltonian and corresponds to the energy observable. 

For example, the Hamiltonian of a free nonrelativistic particle in $\mathbb{R}^n$ is just its kinetic energy and is given by 
$$T=\frac{p^2}{2m}= -\frac{\hbar^2 }{2m} \Delta,$$ 
where $\hbar$ is Planck's constant and $m$ is the mass of the particle (In the following, for simplicity, we will set $\hbar^2/ 2m=1$). $T$ is an operator on the Hilbert space $L^2(\mathbb{R}^n)$ with domain, e.g., $D(T)= C^\infty_c(\mathbb{R}^n)$, the smooth functions of compact support. On $D(T)$ the Laplacian is  symmetric but is not self-adjoint, whence it does not represent any physical observable, and cannot generate any physical dynamics. However, its closure $\overline{T}$, whose domain is the second Sobolev space, \emph{is} a self-adjoint operator (i.e. $T$ is essentially self-adjoint), and as such it is \emph{the} Hamiltonian of a free particle~\cite{RS}. 

The situation drastically changes  in the presence of boundaries. The kinetic energy operator of a particle in an open bounded  set $\Omega$, defined as the Laplace operator on $L^2(\Omega)$ with domain $D(T)= C^\infty_c(\Omega)$, is still symmetric but it is no more essentially self-adjoint, and as such its closure does not correspond to any physical observable. In fact, $T$ admits infinite self-adjoint extensions --- that is infinite possible dynamics --- each one characterized by a given physical behavior of the particle at the boundary $\partial\Omega$.

This paradigmatic example explains the compelling reason, since the inception of quantum mechanics, for searching and characterizing all the self-adjoint extensions (if any) of a given symmetric operator, which is formally associated to a system on physical grounds. And the answer was soon given by von Neumann in his beautiful theory of self-adjoint extensions, which is one of the gems of functional analysis~\cite{vN}. This theory is fully general and completely solve the problem of self-adjoint extensions of every densely defined and closed symmetric operator in an abstract Hilbert space in terms of unitary operators between its deficiency subspaces. See e.g.~\cite{deoliv}. 

In order to reach its goal of encompassing all possible operators, von Neumann's theory should necessarily work at an abstract level. However, for specific classes of operators it would be desirable to have a more concrete characterization of the set of self-adjoint extensions. In particular, for differential operators on a bounded spatial region, as in the above example of the free particle, one would like to establish a direct connection between self-adjoint extensions and boundary conditions. This is highly appealing from a physical perspective, since it would allow  to implement a specific dynamics by building the confining wall  out of a suitable material.

A concrete characterization was given by Grubb~\cite{gg2} for symmetric even-order elliptic differential operators in a  bounded regular spatial domain.
Building on the earlier work of Vi\u{s}ik~\cite{Vishik}, Birman~\cite{bir}, and Lions and Magenes~\cite{lm},  she was able to characterize  all the self-adjoint extensions 
in terms of boundary conditions  parametrized by (unbounded) self-adjoint boundary operators $L:D(L) \subset \mathcal{X} \to \mathcal{X}^\ast$ acting on closed (proper) subspaces~$\mathcal{X}$ of the boundary Hilbert space. See Theorem~\ref{thm:GrubbM} for the Laplace operator.

At an intermediate level of abstraction between Grubb's and von Neumann's descriptions lies the theory of boundary triples~\cite{bgp,deoliv}, which elaborates on ideas of Calkin~\cite{cal} and Vishik~\cite{Vishik}, and is valid for every symmetric operators, because it relies on an abstraction of the notion of boundary values in function spaces. A related description was discovered in the last years by Posilicano~\cite{p}, who introduced  a parametrization in terms of pairs $(\Pi, \Theta)$, where $\Pi$ is an orthogonal projection in an auxiliary Hilbert space $\mathfrak{h}$ and $\Theta$ is a self-adjoint operator in the range of~$\Pi$. See also~\cite{da}. When particularized to differential operators, one recovers Grubb's parametrization, where $\mathfrak{h}$ is essentially the boundary space, $\Pi$ is  the projection onto $\mathcal{X}$ and $\Theta$ is $L$.

Recently,  Asorey, Marmo and Ibort~\cite{marmo,aim2}  proposed on physical ground a different parametrization of the self-adjoint extensions of differential operators in terms of \emph{unitary} operators $U$ on the boundary.  This description relies more directly on physical intuition and in the last years it has been applied to several systems ranging from one dimensional quantum systems with changing boundary conditions~\cite{trotter} or with moving boundaries~\cite{moving walls, phasebox}, to the Aharonov-Bohm effect~\cite{deolivAB}, to field theories~\cite{cas3}, and in particular to the investigation of vacuum fluctuations and the Casimir effect~\cite{cas,cas2}.

The  large use of this description in several physical applications is  also due to its great manageability: the parametrization is in terms of a single unitary operator $U$ on the boundary, instead of a pair $(\mathcal{X}, L)$ composed of a closed subspace $\mathcal{X}$ 
and a self-adjoint operator $L$, which in general is unbounded and thus also needs a domain specification~$D(L)$. Here,  all information is encoded in a single simpler object.

This characterization 
is close in spirit to von Neumann's theory. 
However, it is different in one essential aspect: the unitaries are boundary operators, rather than bulk operators, and as such they are more directly related to  experimental implementations, as discussed above. In fact, along the same lines as~\cite{posvN}, it is possible to connect the two  pictures, and to exhibit an explicit relation between the boundary and the bulk unitaries, but the result will not be very transparent.

Our main objective is to establish 
a characterization of the self-adjoint extensions of an elliptic differential operator  in terms of unitary operators on the boundary.
In this paper we will focus on the above-mentioned paradigmatic model  of the Laplacian in a bounded regular domain, that is the model of a free nonrelativistic particle in a box. 

We will establish, in Theorem~\ref{thm:GG}, a bijection between the set of self-adjoint extensions of the Laplace operator on a bounded regular domain and the set of boundary unitary operators. Each unitary operator is characteristic of a specific boundary condition, that is a relation between the boundary value, $\bm{\gamma} \psi$, of the function $\psi$ and its normal derivative at the boundary, $\bm{\nu}\psi$, (in the sense of  traces).
Recall that the trace of a function in $L^2(\Omega)$ is in $H^{-1/2}(\partial \Omega)$, the Sobolev space of negative fractional order $-1/2$~\cite{lm}. This will be the natural arena of our boundary conditions, that is the boundary Hilbert space the unitaries $U$ will act on. 

The explicit relation, given in Remark~\ref{rem:equivAIM}, reads
$$ \bm{\mu}\psi - i \bm{\gamma} \psi = U ( \bm{\mu}\psi + i \bm{\gamma} \psi), $$
and, in fact, it links the boundary value $\bm{\gamma} \psi$ of the function $\psi$ to the \emph{regular part} $\bm{\mu}\psi$ of its normal derivative $\bm{\nu}\psi$, see Definition~\ref{def:mu}. This is consistent with a different regularity of the boundary values of the function and of its normal derivative: in general their traces belong to different Sobolev spaces, $H^{-1/2}(\partial \Omega)$ and $H^{-3/2}(\partial \Omega)$ respectively,  and cannot be compared. Interestingly enough, the irregular part of the normal derivative plays no role in the boundary conditions; in fact it is not an independent boundary datum, and indeed is completely determined by the trace of the function $\bm{\gamma} \psi$ through the Dirichlet-to-Neumann operator~\cite{gg1}. 

A crucial ingredient in proving that the irregular part of the normal derivative is immaterial to  the boundary conditions is the generalized Green formula, see Definition~\ref{defn:gGauss} and Proposition~\ref{prop:boundaryform}.  It exploits a gauge freedom in Green's second identity:
one can add and subtract an arbitrary boundary self-adjoint operator to the difference of the normal derivatives. This freedom can be used to get rid of the irregular part of the normal derivative and to gain regularity. In other words, the Dirichlet-to-Neumann operator is a self-adjoint operator~\cite{Arendt}.

The link between Grubb's and our parametrization, $(\mathcal{X}, L) \leftrightarrow U$, will be given in Theorem~\ref{thm:relation1-2}. In a few words, the unitary $U$ is adapted to the direct sum $H^{-1/2}(\partial \Omega)= \mathcal{X} \oplus \mathcal{X}^\perp$, and reads $U=V\oplus \mathbb{I}$. Here, the unitary component $V$ is essentially the (partial) Caley transform of $L$ and, as such, it does not have 1 as eigenvalue. Therefore, the eigenspace belonging to the eingenvalue 1 (the idle eigenspace) coincides with $\mathcal{X}^\perp$.

A final remark is in order.  In this paper, for definiteness, we  explicitly consider only the  case of the Laplace operator in a bounded regular domain of $\mathbb{R}^n$. However,
Theorem~\ref{thm:relation1-2} which  establishes the link $(\mathcal{X}, L) \leftrightarrow U$, and the general strategy of encoding boundary conditions in a unitary operator by using an idle subspace and a partial Caley transform, would allow us to generalize our results to a larger class of  operators (e.g.~Laplace-Beltrami~\cite{ibper}, Dirac~\cite{aim2}, pseudodifferential operators~\cite{gg2}) and/or settings (e.g.~manifolds with boundaries~\cite{ibortlledo},  nonregular boundaries~\cite{gesz}).

This article is organized as follows. 
In section~\ref{sec:main}, after setting the notation and defining the regularized normal derivative at the boundary, we state our main result, Theorem~\ref{thm:GG}. Then, after recalling Grubb's characterization of self-adjoint extensions, Theorem~\ref{thm:GrubbM}, we establish the connection betweeen the two parametrizations in Theorem~\ref{thm:relation1-2}. Then we state our result in terms of quadratic forms in Theorem~\ref{thm:form}, which is a corollary of Theorem~\ref{thm:GG}. Sections~\ref{sec:proofs1} and~\ref{sec:proofs2} are devoted to the proofs of the theorems. 
The main properties of the Cayley transform which are used in the proofs are gathered in the final section~\ref{sec:supplemental}.

\section{Notation and Main Results}
\label{sec:main}
We are going to consider complex separable Hilbert spaces. The inner product between two vectors $u$,$v$ of a Hilbert space $\mathscr{H}$ is denoted by $\langle u|v\rangle_{\mathscr{H}}$. In our convention it is anti-linear in the first argument and linear in the second one.

Given two Hilbert spaces $\mathscr{H}_1$ and $\mathscr{H}_2$, the set of unitary operators from $\mathscr{H}_1$ to $\mathscr{H}_2$ is denoted by $\mathscr{U}(\mathscr{H}_1,\mathscr{H}_2) $, while $\mathscr{U}(\mathscr{H}_1$) stands for $\mathscr{U}(\mathscr{H}_1,\mathscr{H}_1)$.

Let $\mathscr{H}$ be an Hilbert space and $A$ a densely defined linear operator on $\mathscr{H}$, 
$$
A : D(A) \subset \mathscr{H} \rightarrow \mathscr{H}.
$$ 
We are going to denote by $A^*$ the adjoint operator of $A$, 
$$
A^* : D(A^*) \subset \mathscr{H}\rightarrow \mathscr{H}.
$$
We say that $A$ is \emph{self-adjoint}  if $A=A^{*}$. 

Let $\Omega$ be an open bounded set in $\mathbb{R}^n$, $n \in \mathbb{N}$. Let the boundary $\partial\Omega$ of $\Omega$ be a $n-1$ dimensional infinitely differentiable manifold, such that $\Omega$ is locally on one side of $\partial\Omega$. From now on a set $\Omega$ satisfying the above conditions will be called a \emph{regular domain}~\cite{lm}. By convention  the normal  $\nu$ of $\partial\Omega$ is oriented towards the exterior of $\Omega$. 

Let $H^s(\Omega)$ (resp. $H^s(\partial\Omega)$), $s\in\mathbb{R}$ , be the Sobolev space of order $s$ on  $\Omega$ (resp. on $\partial\Omega$) with the usual norm~\cite{horm,lm}. Furthermore we set $H^{s}_0(\Omega)$ the closure of $C^\infty_c(\Omega)$ in $H^s(\Omega)$, where $C^\infty_c(\Omega)$ is the set of $C^\infty$ functions with compact support in $\Omega$.  

In what follows we will need the following family of operators $\{\Lambda_t\}_{t\in\mathbb{R}}$, where for all $t \in \mathbb{R}$ the operator $\Lambda_t$ is defined as
$$
\Lambda_t=(\mathbb{I}-\Delta_{\mathrm{LB}})^{t/2}, 
$$
where $\mathbb{I}$ is the identity operator on $L^2(\partial \Omega)$ and $\Delta_{\mathrm{LB}}$ is the Laplace-Beltrami operator on~$\partial\Omega$. We will set $\Lambda\equiv \Lambda_1$. The family $\{\Lambda_t\}_{t\in\mathbb{R}}$ has the following property:   for all $t, s\in \mathbb{R}$
\begin{equation*}
\Lambda_t:H^s(\partial \Omega) \to H^{s-t}(\partial \Omega)
\end{equation*}
is positive and unitary. For an explicit construction of $\{\Lambda_t\}_{t\in\mathbb{R}}$ see~\cite{lm}.

We denote by $\langle \cdot, \cdot\rangle_{s,-s}$, with $s\in\mathbb{R}$, the pairing between $H^{-s}(\partial\Omega)$ and its dual $H^{s}(\partial\Omega)$ induced by the scalar product in $L^2(\partial\Omega)$, i.e.
\begin{equation*}
\langle u, v \rangle_{s,-s}:=\langle \Lambda_s u | \Lambda_{-s} v \rangle_{L^2(\partial\Omega)}, \quad \text{for all } u \in H^s(\partial\Omega), v \in H^{-s}(\partial\Omega).
\end{equation*}
We denote by
\begin{equation*}
\bm{\gamma}  : L^2(\Omega)\to  H^{-1/2}(\partial \Omega), \quad \psi \mapsto \bm{\gamma}(\psi)=\psi|_{\partial \Omega}
\end{equation*}
the \emph{trace operator}, and by
\begin{equation*}
\bm{\nu} : L^2(\Omega) \to H^{-3/2}(\partial \Omega), \quad \psi\mapsto \bm{\nu}(\psi)=\frac{\partial \psi }{\partial \nu} = (\nabla \psi) |_{\partial \Omega} \cdot \nu
\end{equation*}
the \emph{normal derivative} and we recall that these operators are continuous with respect to the respective Hilbert space norms~\cite{lm}. 

In the following we will consider the Laplace operator $T=-\Delta$ on the domain
\begin{equation}
\label{eq:lapl}
D(T)=\{\psi \in H^{2}(\Omega)\,|\,\bm{\gamma}\psi=\bm{\nu}\psi=0\}\equiv H^{2}_0(\Omega),
\end{equation}
and the Dirichlet Laplacian, $T_{\mathrm{D}}=-\Delta$ on 
$$
D(T_{\mathrm{D}})=H^{2}(\Omega)\cap H^{1}_0(\Omega)=\{\psi \in H^{2}(\Omega)\,|\,\bm{\gamma}\psi=0\}.
$$
We recall that $T$ is nothing but the closure in $L^2(\Omega)$ of the symmetric operator given by the Laplacian on functions in $C^\infty_c(\Omega)$. Moreover, $T_{\mathrm{D}}$ is a self-adjoint, positive-definite operator, $T_{\mathrm{D}}=T_{\mathrm{D}}^*>0$.

Let $T^*$ be the adjoint operator of $T$. It acts as the distributional Laplacian on the maximal domain 
$$
D(T^*) =\{\psi \in L^{2}(\Omega)\,|\,\Delta\psi \in L^{2}(\Omega) \}.
$$
We have that the operator $T$ is symmetric, and $T_{\mathrm{D}}$ is a self-adjoint extension of $T$, namely,
$$
T\subset T_{\mathrm{D}} \subset T^*.
$$

Our main objective is to characterize all the possible self-adjoint extensions of the symmetric operator $T$. As the Dirichlet Laplacian, they will all be contained between the minimal Laplacian $T$ and the maximal one $T^*$. The domain of each self-adjoint extension will be characterized by a specific relation between the values of the functions and those of their normal derivatives at the boundary.

In quantum mechanics every self-adjoint extension represents the kinetic energy operator of a free non-relativistic particle (with $\hbar^2/2m=1$), constrained in the spatial domain $\Omega$ by a suitable specific wall.

We will need a regularized  version of the trace operator for the normal derivative $\bm{\nu}$.
\begin{definition}
\label{def:mu}
The  \emph{regularized normal derivative} $\bm{\mu}: D(T^*) \to  H^{-1/2}(\partial\Omega)$ is the linear operator whose action is
\begin{equation*}
\bm{\mu} \psi = \Lambda\, \bm{\nu} \Pi_{\mathrm{D}} \psi, 
\end{equation*}
for all $\psi \in D(T^*)$, where $\Pi_{\mathrm{D}} = T_{\mathrm{D}}^{-1} T^*$.
\end{definition}

\begin{remark}
\label{rem:projdir}
Note that $T_{\mathrm{D}}^{-1}$ maps $L^2(\Omega)$ onto $D(T_{\mathrm{D}})\subset H^2(\Omega)$.  By the trace theorem, $\bm{\nu}(H^2(\Omega)) = H^{1/2}(\partial\Omega)$, whence $\bm{\mu}\psi \in H^{-1/2}(\partial\Omega)$ is more regular than the normal derivative $\bm{\nu}\psi \in  H^{-3/2}(\partial\Omega)$. 

The operator $\Pi_{\mathrm{D}}$ 
is in fact a (nonorthogonal) projection from $D(T^*)$ onto $D(T_{\mathrm{D}})$, since 
for all $\psi \in D(T_{\mathrm{D}})$ one gets that  $\Pi_{\mathrm{D}} \psi = T_{\mathrm{D}}^{-1} T^*\psi = T_{\mathrm{D}}^{-1} T_{\mathrm{D}}\psi = \psi$. Thus, $\bm{\mu} \psi$ is the image under $\Lambda$ of the normal derivative of the component  $\psi_{\mathrm{D}}=\Pi_{\mathrm{D}} \psi$ of $\psi$ belonging to the regular subspace $D(T_{\mathrm{D}})$ of $D(T^*)$.
\end{remark}

\begin{theorem}\label{thm:GG}
The set of all self-adjoint extensions of T is 
\begin{equation*}
\left\{T_U: D(T_U) \to L^2(\Omega) \,|\,U \in \mathscr{U}(H^{-1/2}(\partial \Omega))\right\},
\end{equation*}
where for all $U \in \mathscr{U}(H^{-1/2}(\partial \Omega))$
\begin{equation*}
D(T_U)=\left\{\psi \in D(T^*)\,|\, i (\mathbb{I}+U)\bm{\gamma} \psi =(\mathbb{I}-U) \bm{\mu}\psi \right\}.
\end{equation*}
\end{theorem}

\begin{remark}
\label{rem:equivAIM}
Note the role played by the regularized normal derivative $\bm{\mu} \psi$ in the above theorem:  the trace $\bm{\gamma}\psi$ and $\bm{\mu} \psi$ can be compared because they both belong  to the same (boundary) space, namely $H^{-1/2}(\partial \Omega)$. Notice also the equivalent relation
\begin{equation*}
 \bm{\mu}\psi - i \bm{\gamma} \psi = U ( \bm{\mu}\psi + i \bm{\gamma} \psi)
 \end{equation*}
defining the domain of the self adjoint extension $T_U$. 
\end{remark}

Now we want to compare the result in Theorem~\ref{thm:GG} with the classical characterization of the self-adjoint extensions of $T$ due to Grubb~\cite{gg2,gg1}. 
We need some notation: a closed linear subspace $\mathcal{X}$ of $H^{-1/2}(\partial \Omega)$ is denoted by $ \mathcal{X} \sqsubset H^{-1/2}(\partial \Omega)$, and $\mathcal{X}^\ast$ denotes its dual;
we say that a densely defined operator $L:D(L) \subset \mathcal{X} \to \mathcal{X}^\ast$ is self-adjoint  if 
$$
\Lambda L:D(L) \subset \mathcal{X} \to \mathcal{X}
$$
is self-adjoint, $(\Lambda L)^* = \Lambda L$.

\begin{theorem}[\cite{gg2}]\label{thm:GrubbM}
The set of all self-adjoint extensions of T is 
\begin{equation*}
\left\{ T_{(\mathcal{X},L)}: D\left(T_{(\mathcal{X},L)}\right) \to L^2(\Omega) \,|\, \mathcal{X} \sqsubset H^{-1/2}(\partial \Omega), \; L: D(L) \subset \mathcal{X} \to \mathcal{X}^{\ast},  L \text{ self-adjoint} \right\},
\end{equation*}
where, for all $\mathcal{X} \sqsubset H^{-1/2}(\partial \Omega)$ and $L: D(L) \subset \mathcal{X} \to \mathcal{X}^{\ast}$, $L$ self-adjoint, 
\begin{equation*}
D\left(T_{(\mathcal{X},L)}\right) =\{\psi\in D(T^*)\,|\,\bm{\gamma}\psi\in D(L), \, 
\langle \bm{\nu} \Pi_{\mathrm{D}} \psi, u \rangle_{\frac{1}{2},-\frac{1}{2}} 
= \langle L\bm{\gamma}\psi, u \rangle_{\frac{1}{2},-\frac{1}{2}}, \, \forall u \in \mathcal{X} \}.
\end{equation*}
\end{theorem}
The relation between the two different parametrizations of the self-adjoint extensions of $T$ given in Theorem~\ref{thm:GG} and Theorem~\ref{thm:GrubbM} is established in  the next theorem. First we introduce some notation: if $U: H^{-1/2}(\partial \Omega) \to H^{-1/2}(\partial \Omega)$ is a linear operator and $\mathcal{X}$ is a subspace of $H^{-1/2}(\partial \Omega)$ we denote by $U\!\!\upharpoonright_{\mathcal{X}}$ the operator
$$
U\!\!\upharpoonright_{\mathcal{X}}: \mathcal{X} \to U(\mathcal{X}), \qquad u \in \mathcal{\mathcal{X}} \mapsto U u .
$$

\begin{theorem}\label{thm:relation1-2}
For all $\mathcal{X} \sqsubset H^{-1/2}(\partial \Omega)$ and $L: D(L) \subset \mathcal{X} \to \mathcal{X}^{\ast}$, with $L$ self-adjoint, it results that
\begin{equation*}
T_{(\mathcal{X},L)}=T_U, \,\, \textrm{ with }\; U=\mathscr{C}(\Lambda L) \oplus \mathbb{I}_{\mathcal{X}^{\perp}} \in \mathscr{U} (H^{-1/2}(\partial \Omega)),
\end{equation*}
where 
\begin{equation*}
\mathscr{C}(\Lambda L) = (\Lambda L - i \mathbb{I}_{\mathcal{X}})(\Lambda L + i \mathbb{I}_{\mathcal{X}})^{-1}
\end{equation*}
is the Cayley transform of $\Lambda L$, and $\mathbb{I}_{\mathcal{X}}, \mathbb{I}_{\mathcal{X}^{\perp}}$ are the identity operators on $\mathcal{X}$ and on $\mathcal{X}^{\perp}$, respectively.

Conversely for all $U \in \mathscr{U}(H^{-1/2}(\partial \Omega))$ it results that
\begin{equation*}
T_U=T_{(\mathcal{X},L)},  \,\, \text{with }\; \mathcal{X}=
\mathrm{Ran\,} Q_U \text{ and } L= \Lambda^{-1} \mathscr{C}^{-1}\left( U\!\!\upharpoonright_{\mathcal{X}}  \right),
\end{equation*}
where $Q_U$ 
is the spectral projection of $U$ on the Borel set $\mathbb{R}\setminus \{1\}$ and
$$
\mathscr{C}^{-1}\left( V \right)=i\left(\mathbb{I}_{\mathcal{X}} + V \right)\left(\mathbb{I}_{\mathcal{X}} - V  \right)^{-1}
$$ 
is the inverse Cayley transform of $V \in \mathscr{U}(\mathcal{X})$.
\end{theorem}

\begin{remark}
The Cayley transform maps bijectively self-adjoint operators on the Hilbert space $\mathcal{X}$ to unitary operators that do not have 1 as eigenvalue.
See Section~\ref{sec:supplemental}.  In the second part of the theorem, $V= U\!\!\upharpoonright_{\mathcal{X}}$ is the restriction of the unitary $U$ to its spectral subspace  $\mathcal{X}=\mathrm{Ran\,} Q_U$ orthogonal to the (possible) eigenspace belonging to the eigenvalue~1. Therefore its inverse Cayley transform exists. It is a bounded self-adjoint operator if $1$ is a point of the resolvent set of $V$, i.e.\ if the (possible) eigenvalue 1 of $U$ is isolated; otherwise it is an unbounded self-adjoint operator.
\end{remark}
\section{Quadratic forms}

Consider the expectation value of the symmetric operator $T= -\Delta$ at $\psi\in D(T)= H^{2}_0(\Omega)$:
\begin{equation}
\label{eq:tmin}
\mathfrak{t}(\psi) = \langle \psi |T \psi \rangle_{L^2(\Omega)} = \|\nabla  \psi \|^2_{L^2(\Omega)}.
\end{equation}
Physically this represents the kinetic energy of a quantum particle in the vector state $\psi$ (assumed to be normalized). According to the postulates of quantum mechanics, a quadratic form corresponds to a physical observable---and hence to a self-adjoint operator---if and only if
it is real and closed~\cite{RS}. Therefore, the search of the self-adjoint extensions of the symmetric operator $T$ is mirrored in the search of the real and closed quadratic forms that extend the minimal form~(\ref{eq:tmin}). 

As a consequence, Theorem~\ref{thm:GG}  has a counterpart in terms of kinetic energy forms, through the relation
$\mathfrak{t}_U(\psi) = \langle \psi | T_U  \psi \rangle$,
which must hold for all $\psi\in D(T_U)$.

\begin{theorem}\label{thm:form}
The set of all real closed quadratic forms on $L^2(\Omega)$ that extend $\mathfrak{t}(\psi)$ is
\begin{equation*}
\left\{\mathfrak{t}_U: D(\mathfrak{t}_U) \to \mathbb{R} \,|\,U \in \mathscr{U}(H^{-1/2}(\partial \Omega))\right\},
\end{equation*}
with 
$$\mathfrak{t}_U(\psi)= \|\nabla  \psi_{\mathrm{D}}\|^2_{L^2(\Omega)} + \langle \bm{\gamma}\psi | K_U \bm{\gamma} \psi\rangle_{H^{-1/2}(\partial\Omega)}
,\qquad \text{for all } \psi\in D_U,
$$
where
$$D_U = D(\mathfrak{t}_{\mathrm{D}}) + N(T^*)  \cap \bm{\gamma}^{-1} \left(D(K_U)\right)$$ 
is a core of $\mathfrak{t}_U$.

Here $\psi_{\mathrm{D}} = \Pi_{\mathrm{D}}\psi\in D(\mathfrak{t}_{\mathrm{D}})=H^1_0(\Omega)$, the domain of the Dirichlet form, and
$K_U$ is a self-adjoint operator on the boundary space $H^{-1/2}(\partial\Omega)$ defined by
$$D(K_U)=\mathrm{Ran}(\mathbb{I}-U), \qquad K_U (\mathbb{I}-U)g = - i Q_U  (\mathbb{I}+U)g, \qquad \text{for all } g\in H^{-1/2}(\partial\Omega),$$
with $Q_U$ the  projection onto the subspace $\overline{\mathrm{Ran}(\mathbb{I}-U)}$.

Moreover, the domain $D(T_U)$ of Theorem~\ref{thm:GG} is a core of $\mathfrak{t}_U$ (in fact it is a subspace of $D_U$), and 
$$\mathfrak{t}_U(\psi) = \langle \psi | T_U  \psi \rangle_{L^2(\Omega)} \qquad \text{for all } \psi\in D(T_U).$$

\end{theorem}

\begin{proof}
According to assertion~\ref{thm:directsum} of Lemma~\ref{lemma:regular}, every $\phi\in C^{\infty}(\overline{\Omega})\subset D(T^*)$ has a unique decomposition $\phi = \phi_{\mathrm{D}} + \phi_0$, with $\bm{\gamma}\phi_{\mathrm{D}} = 0$ and $\Delta \phi_0=0$. Thus, for any $\phi\in C^{\infty}(\overline{\Omega})$, we get by the Gauss-Green formula and Definition~\ref{def:mu}
\begin{eqnarray}
\langle \phi|T^{*}\phi\rangle_{L^2(\Omega)}   &=& - \int_\Omega  \bar{\phi} \Delta \phi_{\mathrm{D}} \mathrm{d} x \nonumber \\
& & =  \int_\Omega \nabla \bar{\phi}_0 \cdot \nabla \phi_{\mathrm{D}} \mathrm{d} x 
+\int_\Omega |\nabla \phi_{\mathrm{D}}|^2 \mathrm{d} x 
- \int_{\partial\Omega}  \bar{\phi}  \frac{\partial \phi_{\mathrm{D}}}{\partial \nu} \mathrm{d} S \nonumber \\
& & = \|\nabla  \phi_{\mathrm{D}}\|^2_{L^2(\Omega)} - \langle\bm{\gamma}\phi|  \bm{\mu} \phi\rangle_{H^{-1/2}(\partial\Omega)},
\label{eq:semiGreen}
\end{eqnarray}
since 
\begin{equation*}
\int_\Omega \nabla \bar{\phi}_0 \cdot \nabla \phi_{\mathrm{D}} \mathrm{d} x =  
- \int_\Omega  \Delta\bar{\phi}_0  \phi_{\mathrm{D}} \mathrm{d} x + \int_{\partial\Omega}    \frac{\partial \bar{\phi}_0}{\partial \nu} \phi_{\mathrm{D}} \mathrm{d} S = 0.
\end{equation*}
By density, formula~(\ref{eq:semiGreen}) is valid for all $\phi \in D(T^*)$.
Therefore, we can define the following quadratic form
\begin{equation}
\mathfrak{t}_*(\psi)= \|\nabla  \psi_{\mathrm{D}}\|^2_{L^2(\Omega)} - \langle \bm{\gamma}\psi | \bm{\mu} \psi\rangle_{H^{-1/2}(\partial\Omega)},
\label{eq:tform}
\end{equation}
which  on $D(T^*)$ coincides with the expectation value of the operator $T^*$, namely
$$\mathfrak{t}_*(\psi) = \langle \psi|T^{*}\psi\rangle_{L^2(\Omega)},$$
for all $\psi \in D(T^*)$. However, notice that $D(T_{\mathrm{D}})= H^2(\Omega)\cap H^1_0(\Omega)$ is a dense subspace of $D(\mathfrak{t}_{\mathrm{D}}) = H^1_0(\Omega)$, the domain of the Dirichlet quadratic form, 
$$\mathfrak{t}_{\mathrm{D}}(u) = \|\nabla u\|^2_{L^2(\Omega)}.$$ 
Therefore, the form~(\ref{eq:tform}) can be extended by density to functions 
$$\psi \in D(\mathfrak{t}_{\mathrm{D}}) + N(T^*).$$
[Recall the decomposition of Lemma~\ref{lemma:regular}, $D(T^*)=D(T_{\mathrm{D}})+ N(T^*)$.]

Suppose now that $\psi \in D(T_U) \subset D(T^*)$. Thus,
$$\langle \psi|T_U\psi\rangle_{L^2(\Omega)} = \mathfrak{t}_*(\psi)= \|\nabla  \psi_{\mathrm{D}}\|^2_{L^2(\Omega)} - \langle \bm{\gamma}\psi | \bm{\mu} \psi\rangle_{H^{-1/2}(\partial\Omega)},$$
and, by Theorem~\ref{thm:GG},
$$ i (\mathbb{I}+U)\bm{\gamma} \psi =(\mathbb{I}-U) \bm{\mu}\psi.$$

Let $P_U$ and $Q_U$ be the spectral projections of $U$ on the Borel sets $\{1\}$ and $\mathbb{R}\setminus\{1\}$, respectively ($P_U$ is zero if $1$ is not an eigenvalue of $U$).
Then the above relation is equivalent to
\begin{equation}
\label{eq:projcond}
P_U \bm{\gamma} \psi=0, \qquad  i (\mathbb{I}+U) Q_U\bm{\gamma} \psi =(\mathbb{I}-U) Q_U\bm{\mu}\psi,
\end{equation}
which imply that 
$$\bm{\gamma} \psi \in 
\mathrm{Ran\,}(\mathbb{I}-U) \subset \mathrm{Ran\,}Q_U,$$
since $\mathrm{Ran\,}P_U=  \mathrm{Ran\,}(\mathbb{I}-U)^\perp$.
Let us now define the operator $K_U$ with domain
$$D(K_U)=\mathrm{Ran}(\mathbb{I}-U),$$
whose action is
$$K_U (\mathbb{I}-U)g 
= - i Q_U  (\mathbb{I}+U)g =  - i  (\mathbb{I}+U) Q_U  g,$$ 
for all $g\in H^{-1/2}(\partial\Omega)$.
Thus we get that, for some $g\in H^{-1/2}(\partial\Omega)$,
\begin{eqnarray*}
i (\mathbb{I}+U) Q_U\bm{\gamma} \psi & = & i (\mathbb{I}+U) Q_U (\mathbb{I}-U)g = (\mathbb{I}-U) i Q_U (\mathbb{I}+U)g 
\nonumber\\
&=&- (\mathbb{I}-U) K_U (\mathbb{I}-U)g = - (\mathbb{I}-U) K_U Q_U\bm{\gamma} \psi,
\end{eqnarray*}
which plugged into~(\ref{eq:projcond}) gives
$$ - (\mathbb{I}-U) Q_U K_U \bm{\gamma} \psi = (\mathbb{I}-U) Q_U\bm{\mu}\psi.$$
Since $\mathbb{I}-U$ is injective when restricted to $\mathrm{Ran\,}Q_U$, we get that 
\begin{equation}
 K_U \bm{\gamma} \psi = - Q_U\bm{\mu}\psi,
\end{equation}
for all $\bm{\gamma} \psi \in D(K_U)$.
This implies that
$$- \langle \bm{\gamma}\psi | \bm{\mu} \psi\rangle_{H^{-1/2}(\partial\Omega)} = \langle \bm{\gamma}\psi | K_U \bm{\gamma} \psi\rangle_{H^{-1/2}(\partial\Omega)},$$
for all $\psi \in D(\mathfrak{t}_{\mathrm{D}}) + N(T^*)$, such that $\bm{\gamma} \psi \in D(K_U)$.

Thus we can define the quadratic form
$$\mathfrak{t}_U(\psi)= \|\nabla  \psi_{\mathrm{D}}\|^2_{L^2(\Omega)} + \langle \bm{\gamma}\psi | K_U \bm{\gamma} \psi\rangle_{H^{-1/2}(\partial\Omega)},$$
on the domain
$$D_U=  D(\mathfrak{t}_{\mathrm{D}}) + N(T^*)  \cap \bm{\gamma}^{-1} \left(D(K_U)\right). $$
For all $\psi \in D(T_U)$ it coincides with the expectation value of the self-adjoint extension~$T_U$:
$$\mathfrak{t}_U(\psi) = \langle \psi|T_U\psi\rangle_{L^2(\Omega)}. $$
The domain $D_U$ is a core of the quadratic form $\mathfrak{t}_U$ since it contains the domain of its associated self-adjoint operator $T_U$, namely 
$D(T_U) \subset D_U$.
\end{proof}

\begin{remark}
At variance with the domains of their corresponding operators, the domains of the kinetic energy forms are all contained between a minimal domain and a maximal one:
$$D(\mathfrak{t}_{\mathbb{I}}) 
\subset D(\mathfrak{t}_U) \subset  D(\mathfrak{t}_{-\mathbb{I}}).$$

The Dirichlet form $\mathfrak{t}_{\mathbb{I}}=\mathfrak{t}_{\mathrm{D}}$ has the expression 
$$\mathfrak{t}_{\mathrm{D}}(\psi) = \|\nabla \psi\|^2_{L^2(\Omega)},$$ 
on the minimal domain $D(\mathfrak{t}_{\mathrm{D}}) = H^1_0(\Omega)$, while the form $\mathfrak{t}_{-\mathbb{I}}$ 
has maximal domain  $D(\mathfrak{t}_{-\mathbb{I}}) = H^1_0(\Omega) + N(T^*)$ and acts as
$$\mathfrak{t}_{-\mathbb{I}}(\psi) = \|\nabla \psi_{\mathrm{D}}\|^2_{L^2(\Omega)}.$$
Both forms have no boundary term, since the boundary Hamiltonians are both zero, $K_{\mathbb{I}} = K_{-\mathbb{I}}=0$, but  on the smallest and largest domain, respectively: $D(K_{\mathbb{I}}) = \{0\}$ and $D(K_{-\mathbb{I}}) = H^{-1/2}(\partial\Omega)$. The maximal form $\mathfrak{t}_{-\mathbb{I}}$ corresponds to the Kre\u{i}n-von Neumann extension $T_{-\mathbb{I}}$, whose boundary condition is  the vanishing of the regularized normal derivative, $\bm{\mu}\psi=0$~\cite{Krein}.
\end{remark}

\begin{remark}
Notice that the boundary Hamiltonian  $K_{U}$ is nothing but the inverse partial Cayley transform of the unitary $U$ on its spectral subspace $\mathrm{Ran\,}Q_U= \overline{\mathrm{Ran\,}(\mathbb{I}-U)}$. (In the above proof  $Q_U$ has been identified as the spectral projection of $U$ on the Borel set $\mathbb{R}\setminus\{1\}$). Explicitly, one gets
$$K_U = \mathscr{C}^{-1} (U\!\!\upharpoonright_{\mathrm{Ran\,}Q_U}).$$
The inverse Cayley transform is well defined since the restriction of $U$ has the eigenvalue 1 stripped out. Notice, however, that if $1$ is not an isolated eigenvalue of $U$, then $1$ is not in the resolvent set of $U\!\!\upharpoonright_{\mathrm{Ran\,}Q_U}$, and thus $K_U$ is an unbounded operator.
\end{remark}

\section{Proof of Theorem~\ref{thm:GG}}
\label{sec:proofs1}

We will first need some properties of the regularized normal derivative $\mathbb{\mu}$ and of the projection $\Pi_{\mathrm{D}}$.
\begin{lemma}\label{lemma:regular}
The following properties hold:
\begin{enumerate}
\item 
Let $\bm{\mu}$ be the regularized normal derivative of Definition~\ref{def:mu}, then 
\begin{equation*}
\bm{\mu}: D(T^*) \to  H^{-1/2}(\partial\Omega)
\end{equation*}
is a surjective continuous map with respect to the graph norm.

\item \label{thm:directsum}
The domain of the adjoint $D(T^*)$ is the vector space direct sum of the domain of the Dirichlet Laplacian $T_{\mathrm{D}}$ and the kernel of $T^*$:
$$D(T^*) = D(T_{\mathrm{D}}) + N(T^*), \qquad \psi = \psi_{\mathrm{D}} + \psi_0,$$ 
where $\psi \in D(T^*)$, $\psi_{\mathrm{D}}= \Pi_{\mathrm{D}}\psi \in D(T_{\mathrm{D}})$, and  $\psi_{0}= (\mathbb{I}-\Pi_{\mathrm{D}})\psi \in N(T^*)$.
\item \label{thm:trasesurjective}
The map 
\begin{equation*}
\phi \in D(T^\ast) \mapsto (\bm{\gamma}\, \phi,  \bm{\mu}\, \phi)  \in H^{-1/2}(\partial\Omega) \times H^{-1/2}(\partial\Omega)
\end{equation*}
is surjective.
\end{enumerate}
\end{lemma}

\begin{proof}

\begin{enumerate}
\item
The map $\bm{\mu}$ is continuous as a composition of three continuous maps: $\bm{\mu}=\Lambda \bm{\nu}\Pi_{\mathrm{D}}$, with
$\Lambda : H^{1/2}(\partial\Omega)\to H^{-1/2}(\partial\Omega)$ being unitary,
$$\bm{\nu}:H^2(\Omega) \to H^{1/2}(\partial\Omega)$$
being continuous by the trace theorem, and
$$\Pi_{\mathrm{D}}= T_{\mathrm{D}}^{-1} T^* : D(T^*)\to D(T_{\mathrm{D}})= H^2(\Omega)\cap H^1_0(\Omega)$$ 
being a projection, as pointed out in Remark~\ref{rem:projdir}.

Surjectivity follows from  the surjectivity of the  projection 
$\Pi_{\mathrm{D}}$ 
and from the surjectivity of the map
$$\bm{\gamma}_1=(\bm{\gamma},\bm{\nu}):H^2(\Omega) \to H^{3/2}(\partial\Omega)\times H^{1/2}(\partial\Omega),$$
which implies the surjectivity of its restriction
$$\bm{\gamma}_1:H^2(\Omega) \cap \bm{\gamma}_1^{-1}(\{0\} \times H^{1/2}(\partial\Omega)) \to \{0\} \times H^{1/2}(\partial\Omega),$$
and thus of the map
$$\bm{\nu}:H^2(\Omega) \cap H^1_0(\Omega) \to H^{1/2}(\partial\Omega).$$

\item
For any $\psi \in D(T^*)$ we have $\psi_{\mathrm{D}}= \Pi_{\mathrm{D}}\psi \in D(T_{\mathrm{D}})$ and $\psi_{0}= (\mathbb{I}-\Pi_{\mathrm{D}})\psi \in N(T^*)$. Indeed, 
$$T^* \psi_0= T^*\psi - T^*\Pi_{\mathrm{D}}\psi= T^*\psi- T^* T_{\mathrm{D}}^{-1} T^* \psi = 
T^*\psi- T_{\mathrm{D}} T_{\mathrm{D}}^{-1} T^* \psi = 0.$$

\item 
Since $\Lambda:H^{1/2}(\partial\Omega)\to H^{-1/2}(\partial\Omega)$ is unitary,  the surjectivity of the map 
$$(\bm{\gamma},  \bm{\mu}):D(T^*)  \to H^{-1/2}(\partial\Omega) \times H^{-1/2}(\partial\Omega)$$
is equivalent to the surjectivity of 
$$(\bm{\gamma},  \bm{\nu} \Pi_{\mathrm{D}}):D(T^*)  \to H^{-1/2}(\partial\Omega) \times H^{1/2}(\partial\Omega).$$
By the decomposition of point~\ref{thm:directsum} of the Lemma, we get that for any $\psi \in D(T^*)$, $\psi= \psi_{\mathrm{D}}+\psi_{0}$ with $\bm{\gamma} \psi_{\mathrm{D}}=0$ and $\bm{\nu} \Pi_{\mathrm{D}} \psi_0=\bm{\nu} \Pi_{\mathrm{D}}(\mathbb{I}-\Pi_{\mathrm{D}})\psi =0$. Therefore,
$$(\bm{\gamma},  \bm{\nu}\Pi_{\mathrm{D}})\psi = (\bm{\gamma}\psi_0,\bm{\nu}\psi_{\mathrm{D}}).$$
Therefore, the surjectivity of $(\bm{\gamma},  \bm{\mu})$ follows from the separate surjectivity of the two component maps:
$$\bm{\gamma}: N(T^*)\to  H^{-1/2}(\partial\Omega), \qquad \bm{\nu}: D(T_{\mathrm{D}})\to  H^{1/2}(\partial\Omega).$$
The surjectivity of $\bm{\nu}$ has just been proved in part 1. The surjectivity of $\bm{\gamma}$ is nothing but a classical result~\cite{treves} on the existence of an $L^2(\Omega)$-solution to the  Laplace equation $-\Delta u=0$ for any Dirichlet boundary condition $\bm{\gamma} u =g \in H^{-1/2}(\partial\Omega)$. 
\end{enumerate}
\end{proof} 
Using the regularity result in Lemma~\ref{lemma:regular} we can define the generalized Gauss-Green boundary form.
\begin{definition}\label{defn:gGauss}
We define the generalized Gauss-Green boundary form 
\begin{equation*}
\Gamma:D(T^*)\times D(T^*)\to\mathbb{C}
\end{equation*}
such that for all $\phi, \psi \in D(T^\ast)$
\begin{equation*}
\Gamma(\phi,\psi)=\langle  \bm{\mu} \phi|\bm{\gamma} \psi\rangle_{H^{-1/2}(\partial\Omega)}-\langle\bm{\gamma}\phi|  \bm{\mu} \psi\rangle_{H^{-1/2}(\partial\Omega)}. 
\end{equation*}
\end{definition}
In~\cite{gg2} it was proved the following result.
\begin{proposition}\label{prop:boundaryform}
Let $T$ the operator defined in~(\ref{eq:lapl}) and let $\Gamma$ the generalized Gauss-Green boundary form in Definition~\ref{defn:gGauss}. Then
\begin{equation}
\label{eq:bound}
\Gamma(\phi,\psi)=\langle \phi|T^{*}\psi\rangle_{L^2(\Omega)} - \langle T^{*}\phi|\psi\rangle_{L^2(\Omega)} \quad \textrm{for all $\phi,\psi \in D(T^*)$}. 
\end{equation}
\end{proposition}
\begin{proof}
According to Lemma~\ref{lemma:regular}.\ref{thm:directsum}, every $\phi\in C^{\infty}(\overline{\Omega})\subset D(T^*)$ has a unique decomposition $\phi = \phi_{\mathrm{D}} + \phi_0$, with $\bm{\gamma}\phi_{\mathrm{D}} = 0$ and $\Delta \phi_0=0$. Thus, for any $\phi,\psi\in C^{\infty}(\overline{\Omega})$, we get
\begin{eqnarray*}
& &\langle \phi|T^{*}\psi\rangle_{L^2(\Omega)} - \langle T^{*}\phi|\psi\rangle_{L^2(\Omega)}  = \int_\Omega \left(\Delta \bar{\phi}_{\mathrm{D}} \psi - \bar{\phi} \Delta \psi_{\mathrm{D}} \right)\mathrm{d} x \\
& & = 
\int_{\partial\Omega} \left(\frac{\partial \bar{\phi}_{\mathrm{D}}}{\partial \nu} \psi - \bar{\phi}  \frac{\partial \psi_{\mathrm{D}}}{\partial \nu} \right)\mathrm{d} S
= \langle \bm{ \nu} \phi_{\mathrm{D}},\bm{\gamma} \psi\rangle_{\frac{1}{2},-\frac{1}{2}}-\langle \bm{\gamma}\phi,\bm{\nu} \psi_{\mathrm{D}}\rangle_{-\frac{1}{2},\frac{1}{2}}\\
& & = \langle  \bm{\mu} \phi|\bm{\gamma} \psi\rangle_{H^{-1/2}(\partial\Omega)}-\langle\bm{\gamma}\phi|  \bm{\mu} \psi\rangle_{H^{-1/2}(\partial\Omega)},
\end{eqnarray*}
by the Gauss-Green formula and Definition~\ref{def:mu}. The result follows by density.
\end{proof}

We denote by $\mathscr{H}_b:= H^{-1/2}(\partial \Omega) \oplus H^{-1/2}(\partial \Omega)$.

\begin{definition}\label{defn:btilde}
Let  $\mathcal{W}$ be a subspace of $\mathscr{H}_{\textrm{b}}$. We define the  \emph{$\Gamma$-orthogonal} subspace of $\mathcal{W}$ as
\begin{equation*}
\mathcal{W}^{\dagger}:=\left\{ (u_1,u_2) \in \mathscr{H}_{\textrm{b}} \,\big{|}\, \langle u_2 | v_1 \rangle_{H^{-1/2}(\partial\Omega)} = \langle u_1 | v_2 \rangle_{H^{-1/2}(\partial\Omega)}, \, \forall (v_1,v_2) \in \mathcal{W} \right\}.
\end{equation*}
We say that $\mathcal{W}$ is a \emph{maximally isotropic} subspace if $\mathcal{W}=\mathcal{W}^{\dagger}$.
\end{definition}
\begin{proposition}\label{prop:carasa}
Let $\mathcal{W}$ be a subspace of $\mathscr{H}_{\textrm{b}}$ and let $\tilde{T}$ be the restriction of $T^\ast$ to the domain
$$
D(\tilde{T})=\left\{ \phi \in D(T^\ast) \,|\, (\bm{\gamma} \phi, \bm{\mu} \phi) \in \mathcal{W}  \right\}.
$$
Then $\tilde{T}$ is self-adjoint if and only if $\mathcal{W}$ is a closed maximally isotropic subspace.
\end{proposition}
\begin{proof}
First of all we observe that
$$
\tilde{T} \,\, \textrm{is self-adjoint} \iff  \mathcal{G}(\tilde{T}^\ast)=\mathcal{G}(\tilde{T})
$$
and that $D(\tilde{T}) \subset  D(\tilde{T}^\ast) \subset D(T^\ast)$.
The proof follows immediately by observing that the graph of $\tilde{T}$ reads
\begin{eqnarray*}
\mathcal{G}(\tilde{T})&=& \left\{ (\phi, T^\ast \phi) \,\big{|}\, \phi \in D(\tilde{T})\right\} \\
                         &=&  \left\{ (\phi, T^\ast \phi)  \,\big{|}\, \phi \in D(T^\ast),\,\,\,(\bm{\gamma} \phi,  \bm{\mu} \phi) \in \mathcal{W} \right\},
\end{eqnarray*}
while the graph of $\tilde{T}^\ast$ is 
\begin{eqnarray*}
\mathcal{G}(\tilde{T}^\ast)&=& \left\{ (\phi, T^\ast \phi) \,|\, \phi \in D(\tilde{T}^\ast)\right\} \\
                                  &=& \left\{ (\phi, T^\ast \phi) \,|\, \phi \in D(T^\ast),\,\,\Gamma(\phi,\psi)=0, \,\,\,  \forall \psi \in D(\tilde{T})\right\} \\
                                  &=& \left\{ (\phi, T^\ast \phi) \,|\, \phi \in D(T^\ast), \langle u_1|  \bm{\mu} \phi \rangle_{H^{-1/2}(\partial\Omega)}=\langle \bm{\gamma} \phi | u_2 \rangle_{H^{-1/2}(\partial\Omega)},    \forall (u_1,u_2) \in \mathcal{W} \right\}  \\
                                  &=&  \left\{ (\phi, T^\ast \phi)  \,|\, \phi \in D(T^\ast),\,\,(\bm{\gamma} \phi,  \bm{\mu} \phi) \in \mathcal{W}^{\dagger}\right\},
\end{eqnarray*}
and thus $\mathcal{G}(\tilde{T})=\mathcal{G}(\tilde{T}^\ast)$ iff $\mathcal{W}=\mathcal{W}^{\dagger}$.
\end{proof}

The closed maximally isotropic subspaces are characterized by the following theorem, whose straightforward proof can be found in~\cite{bgp}.
\begin{theorem}\label{thm:bru}
\label{thmform}
A closed subspace  $\mathcal{W}$ of $\mathscr{H}_{\textrm{b}}$ is a maximally isotropic subspace if and only if there exists $U \in \mathscr{U}(H^{-1/2}(\partial\Omega))$ such that
\begin{equation*}
\mathcal{W}=\left\{(u_1,u_2)\in \mathscr{H}_{\textrm{b}} \,|\, i(\mathbb{I}+U)u_1=(\mathbb{I}-U)u_2\right\}.
\end{equation*}
\end{theorem}

We can now conclude. 
\begin{proof}[Proof of Theorem~\ref{thm:GG}]
The proof follows immediately from Proposition~\ref{prop:carasa} and Theorem~\ref{thm:bru}.
\end{proof}
\begin{remark}
The proof of Theorem~\ref{thm:GG} can be translated into the language of boundary triples~\cite{bgp},  by saying that $(\mathscr{H}_{\textrm{b}}, \bm{\gamma},  \bm{\mu})$ is a boundary triple for $T^\ast$. This follows by Proposition~\ref{prop:boundaryform} and by assertion~\ref{thm:trasesurjective} of Lemma~\ref{lemma:regular}.
\end{remark}

\section{Proof of Theorem~\ref{thm:relation1-2}}
\label{sec:proofs2}

\begin{proof}

Let $\mathcal{X} \sqsubset H^{-1/2}(\partial \Omega)$ and $L:D(L) \subset \mathcal{X}  \to \mathcal{X} ^\ast$ a self-adjoint operator. For all $\psi \in D(T^\ast)$ we denote by $(\tilde{\bm{\mu}} \psi) \big{|}_{\mathcal{X}}$ the element of $\mathcal{X}^\ast$ defined as follows:
$$
(\tilde{\bm{\mu}} \psi) \big{|}_{\mathcal{X}} u:=\langle \Lambda_{-1}\bm{\mu} \psi, u \rangle_{\frac{1}{2},-\frac{1}{2}}, \qquad \textrm{for all $u \in \mathcal{X}$},
$$
thus we have that
$$
D(T_{(\mathcal{X} ,L)})=\left\{ \psi \in D(T^*) \,|\, \bm{\gamma} \psi \in D(L),\, (\tilde{\bm{\mu}} \psi) \big{|}_{\mathcal{X}} = L \bm{\gamma} \psi \right\}.
$$
The operator $\Lambda L : D(\Lambda L)\subset \mathcal{X} \to \mathcal{X} $  is self-adjoint, where $D(\Lambda L)=D(L)$. We can define 
$$
V=\mathscr{C}(\Lambda L)=(\Lambda L-i \mathbb{I}_{\mathcal{X}})(\Lambda L+i \mathbb{I}_{\mathcal{X}})^{-1}
$$ 
and by Proposition~\ref{prop:Cayley} in Section~\ref{sec:supplemental} we have that $V \in \mathscr{U}(\mathcal{X} )$. Now  observe that, by assertion~\ref{eqn:Cayleyinv} of Proposition~\ref{prop:Cayley}, we can rewrite $D(T_{(\mathcal{X} ,L)})$ as follows
\begin{equation*}
D(T_{(\mathcal{X} ,L)})=\left\{ \psi \in D(T^*) \,|\, \bm{\gamma} \psi \in D(\Lambda L),\, i (\mathbb{I}_{\mathcal{X}}+V)\bm{\gamma} \psi= (\mathbb{I}_{\mathcal{X}}-V)\Lambda (\tilde{\bm{\mu}} \psi) \big{|}_{\mathcal{X}} \right\}.
\end{equation*}
For all $\psi \in D(T^\ast)$ we denote by $( \bm{\mu} \psi) \big{|}_{\mathcal{X}}$ the element of $\mathcal{X}$ defined as follows:
$$
v( \bm{\mu} \psi) \big{|}_{\mathcal{X}}:=\langle v, \bm{\mu} \psi \rangle_{\frac{1}{2},-\frac{1}{2}}, \qquad \textrm{for all $v \in \mathcal{X}^\ast$},
$$
Observe that 
$$
\Lambda (\tilde{\bm{\mu}} \psi) \big{|}_{\mathcal{X}} =(\bm{\mu} \psi) \big{|}_{\mathcal{X}} \quad \textrm{for all $\psi \in D(T^\ast)$},
$$
therefore  $D(T_{(\mathcal{X} ,L)})$ can be rewritten as
\begin{equation*}
D(T_{(\mathcal{X} ,L)})=\left\{\psi \in D(T^*)\,|\,  \bm{\gamma} \psi \in D(\Lambda L),\, i (\mathbb{I}_{\mathcal{X}}+V)\bm{\gamma} \psi= (\mathbb{I}_{\mathcal{X}}-V)( \bm{\mu} \psi) \big{|}_{\mathcal{X}} \right\}.
\end{equation*}
By Lemma~\ref{thm1} in Sec.~\ref{sec:supplemental}, one gets that the condition $\bm{\gamma} \psi \in D(\Lambda L)$ can be dispensed with. Indeed, as long as $\bm{\gamma}\psi \in \mathcal{X} $ satisfies the equation
$$
 i (\mathbb{I}_{\mathcal{X}}+V)\bm{\gamma} \psi =(\mathbb{I}_{\mathcal{X}}-V)( \bm{\mu} \psi) \big{|}_{\mathcal{X}},
$$
then $\bm{\gamma} \psi \in D(\Lambda L)$. 
Therefore we have proved that
\begin{equation*}
D(T_{(\mathcal{X} ,L)})=\left\{\psi \in D(T^*) \,|\, \bm{\gamma} \psi \in \mathcal{X} ,\, i (\mathbb{I}_{\mathcal{X}}+V)\bm{\gamma} \psi= (\mathbb{I}_{\mathcal{X}}-V)( \bm{\mu} \psi) \big{|}_{\mathcal{X}} \right\}.
\end{equation*}
Thus, by defining the operator $U:=V \oplus \mathbb{I}_{{\mathcal{X}}^\perp} \in \mathscr{U}(H^{-1/2}(\partial \Omega))$, we have that $D(T_{(\mathcal{X} ,L)})=D(T_U)$, and that $T_{(\mathcal{X} ,L)}=T_U$ with $U:=\mathscr{C}(\Lambda L) \oplus \mathbb{I}_{{\mathcal{X}}^\perp} $.

Now we prove the converse. Fix $U \in \mathscr{U}(H^{-1/2}(\partial \Omega))$ and consider $T_U$, a self-adjoint extension of $T$. Let $P_U$ the spectral projection of $U$ on the Borel set $\{1\} \subset \mathbb{R}$. Define $\mathcal{X} :=\textrm{Ran}(P_U)^\perp \sqsubset H^{-1/2}(\partial \Omega)$ and  consider the operator $V= U\!\!\upharpoonright_{\mathcal{X}}  \in \mathscr{U}(\mathcal{X})$. Clearly, $1$~is not an eigenvalue of $V$, therefore we can define the self-adjoint operator
$$
L:= \Lambda^{-1}\left[ i(\mathbb{I}_{\mathcal{X}}+V)(\mathbb{I}_{\mathcal{X}}-V)^{-1}\right]:D(L) \subset \mathcal{X}  \to \mathcal{X} ^\ast
$$
We know that 
\begin{equation*}
D(T_U)=\left\{ \psi \in D(T^\ast) \,|\, i(\mathbb{I}+U)\bm{\gamma} \psi= (\mathbb{I}-U) \bm{\mu} \psi \right\}
\end{equation*}
By projecting on $\mathcal{X} $ and ${\mathcal{X}}^{\perp}$ the equation $ i(\mathbb{I}+U)\bm{\gamma} \psi= (\mathbb{I}-U) \bm{\mu} \psi$,  one gets
\begin{equation*}
D(T_U)=\left\{ \psi \in D(T^\ast) \,|\,  \,\bm{\gamma}\psi\in \mathcal{X},\,  i(\mathbb{I}+V)\bm{\gamma} \psi= (\mathbb{I}-V)( \bm{\mu} \psi) \big{|}_{\mathcal{X}}\right\}.
\end{equation*}
Since $( \bm{\mu} \psi) \big{|}_{\mathcal{X}}=\Lambda (\tilde{\bm{\mu}} \psi) \big{|}_{\mathcal{X}}$, for all $\psi \in D(T^\ast)$, we have that
\begin{equation*}
D(T_U)=\left\{ \psi \in D(T^\ast) \,|\, \bm{\gamma}\psi\in \mathcal{X},\,  i(\mathbb{I}+V)\bm{\gamma} \psi= (\mathbb{I}-V)\Lambda (\tilde{\bm{\mu}} \psi) \big{|}_{\mathcal{X}}\right\}.
\end{equation*}
Again by Lemma~\ref{thm1}, one has that 
\begin{equation*}
D(T_U)=\left\{ \psi \in D(T^\ast) \,|\, \bm{\gamma}\psi\in D(\Lambda L),\,  i(\mathbb{I}+V)\bm{\gamma} \psi= (\mathbb{I}-V)\Lambda (\tilde{\bm{\mu}} \psi) \big{|}_{\mathcal{X}}\right\}
\end{equation*}
and thus
\begin{equation*}
D(T_U)=D(T_{(\mathcal{X},L)}). \qedhere
\end{equation*}
\end{proof}

\section{Supplemental results}
\label{sec:supplemental}

Let us recall some basic facts about the Cayley transform of self-adjoint operators. For further details see \cite{rudin}.
\begin{definition}
\label{def:Cayley}
Let $A : D(A)\subset \mathscr{H} \rightarrow \mathscr{H}$ be a self-adjoint operator. We define the \emph{Cayley transform} of $A$, denoted by $\mathscr{C}(A)$, as follows
\begin{equation*}
\mathscr{C}(A)=(A- i \mathbb{I}) (A+ i \mathbb{I})^{-1},
\end{equation*}
where $\mathbb{I}$ is the identity operator on $\mathscr{H}$. 

Conversely, let $U \in \mathscr{U}(\mathscr{H})$ and assume that  $1$ is not an eigenvalue of $U$.  We define the \emph{inverse Cayley transform} of $U$, denoted by $\mathscr{C}^{-1}(U)$, as follows
\begin{equation*}
\mathscr{C}^{-1}(U)=i (\mathbb{I}+U)(\mathbb{I}-U)^{-1}.
\end{equation*}
\end{definition}
\begin{proposition}[\cite{rudin}]
\label{prop:Cayley}
Let 
$A : D(A)\subset \mathscr{H} \rightarrow \mathscr{H}$ be a self-adjoint operator. Then
\begin{enumerate}
\item $\mathscr{C}(A) \in \mathscr{U}(\mathscr{H})$;
\item $\mathbb{I}-\mathscr{C}(A)$ is injective;
\item $\mathrm{Ran}(\mathbb{I}-\mathscr{C}(A))=D(A)$;
\item For all $\phi \in D(A)$,
\begin{equation*}\label{eqn:Cayleyinv}
A\phi=\,i (\mathbb{I}+\mathscr{C}(A))(\mathbb{I}-\mathscr{C}(A))^{-1}\phi = \mathscr{C}^{-1}(\mathscr{C}(A))\phi . 
\end{equation*}
\item
Moreover if $U \in \mathscr{U}(\mathscr{H})$ such that $1$ is not an eigenvalue of $U$ then 
$$\mathscr{C}^{-1}(U): \mathrm{Ran}(\mathbb{I}-U)\to \mathscr{H}$$ 
is a self-adjoint operator and $\mathscr{C}(\mathscr{C}^{-1}(U))=U$.
\end{enumerate}
\end{proposition}

\begin{lemma}
\label{thm1}
Let $A: D(A) \subset \mathscr{H} \to \mathscr{H}$  be a self-adjoint operator and
$$
\mathcal{G}(A)=\{(u,Au)\in \mathscr{H}\times \mathscr{H}\,|\, u\in D(A)\}
$$
be its graph.
Let
\begin{equation*}
\Theta(A)=\{(\phi,\psi)\in \mathscr{H}\times \mathscr{H}\,|\, i (\mathbb{I}+\mathscr{C}(A))\,\phi=(\mathbb{I}-\mathscr{C}(A))\,\psi \}.
\end{equation*}
Then $\mathcal{G}(A)=\Theta(A)$.
\end{lemma}
\begin{proof}
Notice first that the inclusion $\mathcal{G}(A)\subset\Theta(A)$ follows immediately from property~\ref{eqn:Cayleyinv} of Proposition~\ref{prop:Cayley}.
 We need to show that $\Theta(A)\subset \mathcal{G}(A)$. 
 
 Fix $(\phi,\psi)\in \mathscr{H}\times\mathscr{H}$ such that
\begin{equation}
\label{eq:eq1}
 i (\mathbb{I}+\mathscr{C}(A))\,\phi=(\mathbb{I}-\mathscr{C}(A))\,\psi.
\end{equation}
Observe that
$$
\mathbb{I}-\mathscr{C}(A)=2i (A+i \mathbb{I})^{-1} \quad \textrm{and} \quad \mathbb{I}+\mathscr{C}(A)=2 A (A+i \mathbb{I})^{-1}.
$$
Plugging these expressions in equation~(\ref{eq:eq1}) we obtain:
\begin{equation*}
A\,(A+i\,\mathbb{I})^{-1}\phi = (A+i\,\mathbb{I})^{-1}\psi.
\end{equation*}
The right hand side belongs to $D(A)$, thus also the left hand side belongs to $D(A)$. Then we can multiply both sides by $A+i\,\mathbb{I}$, obtaining
\begin{equation*}
\psi=(A+i\,\mathbb{I})\,A\,(A+i\,\mathbb{I})^{-1}\phi.
\end{equation*} 
Since $(A+i\,\mathbb{I})^{-1}\phi \in D(A)$ and $A\,(A+i\,\mathbb{I})^{-1}\phi \in D(A)$, it follows that $(A+i\,\mathbb{I})^{-1}\phi \in D(A^2)$. The operators $(A+i\,\mathbb{I})$ and $A$ commute on $D(A^2)$ and we get that
\begin{equation*}
\psi=A \phi,
\end{equation*}
and thus $(\phi, \psi) \in \mathcal{G}(A)$.
\end{proof}

\section*{Acknowledgments}

We thank Manolo Asorey and Beppe Marmo for useful discussions. This work was  supported by Cohesion and Development Fund 2007-2013 - APQ Research Puglia Region ``Regional program supporting smart specialization and social and environmental sustainability - FutureInResearch'', by the Italian National Group of Mathematical Physics (GNFM-INdAM), and by Istituto Nazionale di Fisica Nucleare (INFN) through the project ``QUANTUM''.

\bibliography{Myrefs}{}
\bibliographystyle{plain}

\end{document}